\documentclass[11pt,envcountsect,envcountsame]{article}
\usepackage[usenames,dvipsnames]{color}
\usepackage{amssymb,amsmath,eepic,graphics,color,enumerate,eucal}
\usepackage{tikz}
\usepackage{xspace}
\usetikzlibrary{calc}
\usepackage{todonotes}
\usepackage[utf8x,utf8]{inputenc} 
\usepackage[T1]{fontenc}

\usepackage[margin=1in]{geometry}

\usepackage{epsfig}
\usepackage{amsmath,amsthm,amsfonts,amssymb,epsfig,color,xy}

\usepackage{comment}

\usepackage{hyperref}
\usepackage[nottoc,notlot,notlof]{tocbibind}
\usepackage[font=small, labelfont=bf, width=0.9\textwidth]{caption}
\usepackage{authblk}

\newtheorem{theorem}{Theorem}[]
\newtheorem{lemma}[theorem]{Lemma}
\newtheorem{corollary}[theorem]{Corollary}

\newtheorem{example}{Example}

\newcommand{\myspan}{{\rm span}}

\newcommand{\ignore}[1]{}%

\newcommand{\ProblemFormat}[1]{{\sc #1}}
\newcommand{\ProblemName}[1]{\ProblemFormat{#1}\xspace}

\newcommand{\CA}{\ProblemName{Channel Assignment}}

\newcommand{\ThreeSAT}{\ProblemName{3-CNF-SAT}}
\newcommand{\TwoSAT}{\ProblemName{2-CNF-SAT}}

\newcommand{\FInt}{\ProblemName{Family Intersection}}
\newcommand{\CMW}{\ProblemName{Common Matching Weight}}

\newcommand{\heading}[1]{\medskip\noindent{\bf #1.\ }}%

\newcommand{\defproblem}[3]{
  \vspace{1mm}
\noindent\fbox{
  \begin{minipage}{0.96\textwidth}
  #1 \\
  {\bf{Input:}} #2  \\
  {\bf{Question:}} #3
  \end{minipage}
  }
}

\newcommand{\br}[1]{\left\{#1\right\}}
\newcommand{\bk}[1]{\left[#1\right]}
\newcommand{\bt}[1]{\left\langle#1\right\rangle}
\newcommand{\p}[1]{\left(#1\right)}

\newcommand{\N}{\mathbb{N}}
\newcommand{\Z}{\mathbb{Z}}
\newcommand{\of}[1]{\p{#1}}
\newcommand{\set}[1]{\br{#1}}
\newcommand{\abs}[1]{\left|#1\right|}
\newcommand{\pow}[1]{\left|#1\right|}
\newcommand{\el}[1]{\bk{#1}}

\newcommand{\sseq}{\subseteq}
\newcommand{\poly}[1]{\mbox{poly}\of{#1}}

\newcommand{\fel}[2]{\el{#1}\to\el{#2}}
\newcommand{\feln}[2]{\el{#1}\times\el{#2}\to\N}

\newcommand{\ndivlog}[1]{\frac{#1}{\log{#1}}}

\newcommand{\rnlogkn}[2]{
  \frac{#1}{\log_{#2}{\frac{#1}{\log_{#2}{\frac{#1}{\ldots}}}}}
}
\newcommand{\seq}[2]{{#1}_1, {#1}_2, \ldots, {#1}_{#2}}

\begin{document}
\title{Tight lower bound for the channel assignment problem}
\author{Arkadiusz Socała
  \thanks{E-mail: \texttt{as277575@students.mimuw.edu.pl}}}
\affil{University of Warsaw, Poland}

\date{}

\setcounter{page}{0}

\maketitle
\vspace{-1cm}
\begin{abstract}
We study the complexity of the \CA problem.
A major open problem asks whether \CA admits an $O(c^n)$-time algorithm,
for a constant $c$ independent of the weights on the edges.
We answer this question in the negative i.e. we show
that there is no
$2^{o(n\log n)}$-time algorithm solving \CA unless the
Exponential Time Hypothesis fails.
Note that the currently best known algorithm works in time
$O^*(n!) = 2^{O(n\log n)}$ so our lower bound is tight.
\end{abstract}

\newpage

\section{Introduction}

In the \CA problem, we are given a symmetric weight function $w:V^2\to\N$
(we assume that $0\in\N$).
The elements of $V$ will be called vertices (as $w$ induces a graph on the
vertex set $V$ with edges corresponding to positive values of $w$).
We say that $w$ is $\ell$-bounded when for every $x, y\in V$ we have
$w(x, y) \le \ell.$
An assignment $c:V\to\Z$ is called {\em proper} when for each pair of vertices
$x, y$ we have $|c(x) - c(y)|\ge w(x, y).$
The number $\left(\max_{v\in V} c(v) - \min_{v\in V} c(v) + 1\right)$ is called
the {\em span} of $c.$
The goal is to find a proper assignment of minimum span.
Note that the special case when $w$ is $1$-bounded corresponds to the classical
graph coloring problem.
It is therefore natural to associate the instance of the channel assignment
problem with an edge-weighted graph $G=(V, E)$ where
$E=\set{uv\ :\ w(u, v) > 0}$ with edge weights $w_E:E\to\N$ such that
$w_E(xy) = w(x, y)$ for every $xy\in E$ (in what follows we abuse the notation
slightly and use the same letter $w$ for both the function defined on $V^2$
and $E$).
The minimum span is called also the span of $(G, w)$ and denoted by
$\myspan(G, w).$

It is interesting to realize the place \CA in a kind of hierarchy of constraint
satisfaction problems.
We have already seen that it is a generalization of the classical graph
coloring.
It is also a special case of the constraint satisfaction problem (CSP).
In CSP, we are given a vertex set $V,$ a constraint set $\mathcal{C}$ and a
number of colors $d.$
Each constraint is a set of pairs of the form $(v, t)$ where $v\in V$ and
$t\in\set{1, \ldots, d}.$
An assignment $c:V\to\set{1, \ldots, d}$ is {\em proper} if every constraint
$A\in\mathcal{C}$ is satisfied, i.e.\ there exists $(v, t)\in A$ such that
$c(v)\ne t.$
The goal is to determine whether there is a proper assignment.
Note that \CA corresponds to CSP where $d=s$ and every edge $uv$ of weight
$w(uv)$ in the instance of \CA corresponds to the set of constraints of the
form $\set{(u, t_1), (v, t_2)}$ where $|t_1 - t_2|<w(uv).$

In the general case th best known algorithm runs in $O^*(n!)$ time
(see McDiarmid~\cite{mcdiarmid}).
However, there has been some progress on the $\ell$-bounded variant.
McDiarmid~\cite{mcdiarmid} came up with an $O^*((2\ell+1)^n)$-time algorithm
which has been next improved by Kral~\cite{kral} to $O^*((\ell+2)^n),$
further to $O^*((\ell+1)^n)$ by Cygan and Kowalik~\cite{cygan}
and to $O^*((2\sqrt{\ell + 1})^n)$ by Kowalik and Socala~\cite{meet}.
These are all dynamic programming (and hence exponential space) algorithms.
The last but one applies the fast zeta transform to get a minor speed-up
and the last one uses the meet-in-the-middle approach.
Interestingly, all these works show also algorithms which {\em count} all
proper assignments of span at most $s$ within the same running time
(up to polynomial factors) as the decision algorithm.

Since graph coloring is solvable in time $O^*(2^n)$~\cite{bhk:coloring}
it is natural to ask whether \CA is solvable in time $O^*(c^n),$ for some
constant $c.$
It is a major open problem (see~\cite{kral,cygan,dagstuhl}) to find such a
$O(c^n)$-time algorithm for $c$ independent of $\ell$ or prove that it does not
exist under a reasonable complexity assumption.
A complexity assumption commonly used in such cases is the Exponential Time
Hypothesis (ETH),
introduced by Impagliazzo, Paturi and Zane~\cite{PZ01}.
It states that \ThreeSAT cannot be computed in time $2^{o(n)},$ where $n$ is
the number of variables in the input formula.
The open problem mentioned above becomes even more interesting when we realize
that under ETH, CSP does not have a $O^*(c^n)$-time algorithm for a constant
$c$ independent of $d,$ as proved by Traxler~\cite{traxler}.

\heading{Our Results}
Our main result is a proof that \CA does not admit a $O(c^n)$-time for a
constant $c$ under the ETH assumption.
By applying a sequence of reductions (see Figure~\ref{fig:sequence})
starting in \ThreeSAT and ending in \CA
we were able to solve this open problem and to show that there is no
$2^{o(n\log n)}$-time algorithm solving \CA unless the ETH fails.
Note that the currently best known algorithm works in time
$O^*(n!) = 2^{O(n\log n)}$ so our lower bound is tight.
\newpage
\heading{\CMW as a generic problem without $2^{o(n\log n)}$-time algorithm}\\
In order to prove that there is no $2^{o(n\log n)}$-time algorithm for some
problem we may want to use a reduction from some better studied problem,
say from \ThreeSAT for which we know that there is no $2^{o(n)}$-time
algorithm unless the ETH fails.
Therefore
in this case we need to be able to transform an instance of \ThreeSAT of
size $n$ into an instance of our target problem of size $O\of{\ndivlog{n}}.$
Then $2^{o(n\log n)}$-time algorithm for our target problem would imply
$2^{o(n)}$-time algorithm for \ThreeSAT which contradicts the ETH.
However such reductions which compress the size of the instance from
$O(n)$ to e.g. $O\of{\frac{n}{\log n}}$ are very rare.
As shown in the Figure~\ref{fig:sequence} we do this for the problem \CMW
defined as follows:

\defproblem{\CMW}{
  Two complete weighted bipartite graphs $G_1=(V_1\cup W_1, E, w_1)$
  and $G_2=(V_2\cup W_2, E, w_2)$ such that $\pow{V_1} = \pow{W_1}$
  and $\pow{V_2} = \pow{W_2}.$
  The weight functions $w_1, w_2$ have nonnegative integer values.
}{
  Are there two perfect matchings $M_1$ in $G_1$ and $M_2$ in $G_2$
  such that $w_1\of{M_1} = w_2\of{M_2}?$
}

\newcommand{\stack}[2]{\begin{tabular}{c} #1 \\ #2 \\ \end{tabular}}

\begin{figure}[t]
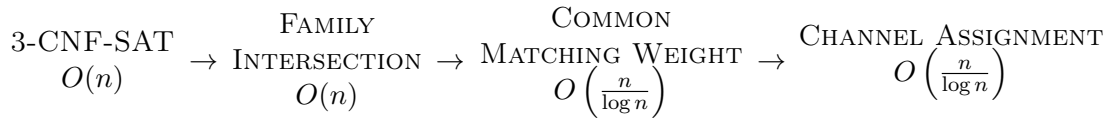

\begin{center}
    \begin{minipage}{0.14\textwidth}
      \begin{center}
        \ThreeSAT\\
        $O(n)$
      \end{center}
    \end{minipage}
    $\rightarrow$
    \begin{minipage}{0.16\textwidth}
      \begin{center}
        \FInt\\
        $O(n)$
      \end{center}
    \end{minipage}
    $\rightarrow$
    \begin{minipage}{0.22\textwidth}
		  \begin{center}
        \CMW\\
        $O\of{\ndivlog{n}}$
		  \end{center}
    \end{minipage}
    $\rightarrow$
    \begin{minipage}{0.25\textwidth}
		  \begin{center}
        \CA\\
        $O\of{\ndivlog{n}}$
      \end{center}
    \end{minipage}
  \caption{
    The sequence of the used reductions and the size of the instance.
    The compression follows between \FInt and \CMW.
    While the definition of \FInt is rather technical
    the \CMW problem is quite natural and it can be used as a generic problem
    without $2^{o(n\log n)}$-time algorithm.
  }
  \label{fig:sequence}
\end{center}
\vspace{-1cm}
\end{figure}

Note that in order to show that a new problem \ProblemName{P} does not
admit a $2^{o(n\log n)}$-time algorithm it suffices to give a linear reduction
from \CMW to \ProblemName{P}.
We shown such reduction for \CA and we hope that the same thing can be done
also for other problems.

\heading{Organization of the paper}
In Section~\ref{sec:CMW} we describe a sequence of reductions
starting in \ThreeSAT and ending in \CMW and the conclusions from
the existence of these reductions leading to the theorem on the hardness
of \CMW.
In Section~\ref{sec:CA} we present a reduction from \CMW to \CA and prove
the hardness of \CA

\heading{Notation}
Throughout the paper $n$ denotes the number of the vertices of the graph under
consideration.
For an integer $k,$ by $\el{k}$ we denote the set $\set{1, 2, \ldots, k}.$

\section{Hardness of \CMW}
\label{sec:CMW}
In this section
we describe a sequence of reductions starting in \ThreeSAT and ending in \CMW
and the consequences of these reductions on the complexity of \CMW.
In the second of these two reductions we compress the instance from the size
$O(n)$ to the size $O\of{\ndivlog{n}}$ which is an important part of our
result.

\subsection{From \ThreeSAT to \FInt.}
\label{subsec:ThreeSAT-to-FInt}

The intuition is that for a given instance of \ThreeSAT we consider a set
of the \emph{occurrences} of the variables in the formula
i.e. we treat any two different
occurrences of the same variable as they were two different variables.
Note that in a \ThreeSAT instance with $n$ variables and $m$ clauses we have
$3m$ occurrences of the $n$ variables so there are $2^{3m}$ assignments of the
occurrences.

We would like to represent two useful subsets of the set of all $2^{3m}$
assignments of the occurrences.
The first is the set of the \emph{consistent} assignments i.e.
such assignments of the occurrences that
all the occurrences of the same variable have the same value.
The second is the set of the assignments of the occurrences such that every
clause is satisfied
(although they are allowed to have different values for different occurrences
of the same variable i.e. they do not need to be consistent).
Note that the instance of \ThreeSAT is a YES-instance if and only if
the intersection of these two sets is nonempty.

To represent those two sets we would like to use the following concept.
For a function $f:\feln{a}{b}$ we define
$X_f = \set{\sum_{i=1}^{a} f\of{i, \sigma\of{i}} | \sigma: \fel{a}{b}}.$
We call this set an \emph{$f$-family}.

We will define a function $f$ such that the elements of the $f$-family $X_f$
correspond to the assignments of the occurrences
such that every two occurrences of the same variable have the same value.
Then we define another function $g$ such that the
$g$-family $X_g$ represents the assignments of the
occurrences such that every clause is satisfied.
Thus we reduce \ThreeSAT into the following problem:

\defproblem{\FInt}{
  A function $f:\feln{a}{b}$ and a function $g:\feln{c}{d}.$
}{
  Does $X_f \cap X_g \neq \emptyset$?
}\\

\begin{example} \label{exmp:TwoSAT-to-FInt}
Let us illustrate our approach on a \TwoSAT formula
$\varphi = (\alpha \vee \beta) \wedge (\neg \alpha \vee \gamma).$
We can make a distinction between different occurrences of the same variables
$\varphi' = (\alpha_1 \vee \beta_1) \wedge (\neg \alpha_2 \vee \gamma_1).$
So we have four occurrences of the variables and $2^4=16$ assignments.
We represent those assignments as
numbers from the set $\set{0, 1, \ldots, 2^4 - 1}$
However, it will be convenient to refer to these numbers as bit vectors
of length $4$ where
the $i$-th bit represents the value of the $i$-th occurrence
(among all the occurrences of all the variables).

To represent the set of the consistent assignments of the occurrences we can
use a function $f:\feln{n}{2}$ such that the value of $f(i, 1)$ is a bit vector
representing all the occurrences of the $i$-th variable and $f(i, 2) = 0.$
So in our example we have
$f(1, 1) = 1010_2,$ $f(2, 1) =  0100_2, $ and $f(3, 1) = 0001_2.$
Therefore
$X_f = \set{0000_2, 0001_2, 0100_2, 0101_2, 1010_2, 1011_2, 1110_2, 1111_2}.$

To represent the set of the assignments which satisfies all the clauses we
can use a function $g:\feln{m}{3}$ such that $g(i, j)$ is a $j$-th assignment
(in some fixed order) 
of the occurrences of the variables in the $i$-th clause which satisfies
this clause.
Note that every clause in \TwoSAT have $3$ assignments of the occurrences
which satisfies this clause.
So in our example we have $g(1, 1) = 1000_2,$ $g(1, 2) = 0100_2,$
$g(1, 3) = 1100_2,$ $g(2, 1) = 0000_2,$ $g(2, 2) = 0001_2$
and $g(2, 3) = 0011_2.$
Therefore
$X_g = \set{0100_2, 0101_2, 0111_2, 1000_2, 1001_2, 1011_2, 1100_2, 1101_2,
1111_2}.$

The set $X_f \cap X_g = \set{0100_2, 0101_2, 1011_2, 1111_2}$ is the set
of all the consistent assignments of the occurrences such that each clause
is satisfied.
\end{example}

We can formalize our observation as following.

\begin{lemma} \label{lem:ThreeSAT-to-FInt} ($\bigstar$)
There is a polynomial time reduction from a given instance of \ThreeSAT with
$n$ variables and $m$ clauses
into an instance of \FInt with $f:\feln{n}{2}$ and $g:\feln{m}{7}$
such that $\max{X_f} < 2^{3m}$ and  $\max{X_g} < 2^{3m}.$
\end{lemma}

The proof is straightforward and its idea should be illustrated
by the Example~\ref{exmp:TwoSAT-to-FInt}.
It is moved to the Appendix due to space limitations.

\subsection{From \FInt to \CMW.}
\label{subsec:FInt-to-CMW}

Consider an $f$-family $X_f$
and a $g$-family $X_g$ for some functions $f:\feln{n}{2}$ and $g:\feln{m}{7}.$
In this section we show how to encode $X_f$ in some weighted bipartite
graph $G_1$ so that the set of the weights of the perfect matchings in $G_1$
will be equal to $X_f.$
Similarly we will encode $X_g$ in some bipartite graph $G_2$
such that the set of the weights of the perfect matchings in $G_2$
will be equal to $X_g.$
So the set $X_f \cap X_g$ is nonempty if and only if
$G_1$ and $G_2$ contain perfect matchings with the same weight.
Moreover the number of the vertices of the graph $G_1$ will be
$O\of{\ndivlog{n}}$ and the number of the vertices of the graph $G_2$
will be $O\of{\ndivlog{m}}.$
This is a crucial step of our construction, because the instance size 
decreases (by a logarithmic factor).

\newcommand{\h}[1]{\hat{#1}}
\newcommand{\hN}{\h{\N}}
\newcommand{\hel}[1]{\h{\el{#1}}}

Before we describe the reduction we need the following technical lemma which
describe constructions of permutations of some specified
properties.
The permutations correspond naturally to perfect matchings in bipartite graphs.
Elements of $\el{k}^b$ will be treated as $b$-character words over alphabet
$\el{k},$ i.e. for $x\in\el{k}$ and $w\in[k]^b$ by $xw$ we mean the word of
length $b + 1$ obtained by concatenating $x$ and $w.$
For convenience we define a set $\hN = \set{\h{0}, \h{1}, \h{2}, \ldots}$
as a copy of the natural numbers $\N$
and for every $n\in\N$ we define $\hel{n} = \set{\h{1}, \h{2}, \ldots, \h{n}}.$
Every set $\hel{k}^b$ is just a copy of $\el{k}^b$
so we refer to bijections between $\el{k}^b$ and $\hel{k}^b$
as to permutations.

\begin{lemma} \label{lem:bijection} ($\bigstar$)
Let $b\in \N$ and $\alpha:\hel{k}^b \times \el{b} \to \el{k} \cup \set{\bot}$
such that for every $\h{w}\in\hel{k}^b$ and for every $i\in\el{b}$
holds $\alpha\of{\h{w}, i} \neq \bot$ if and only if $\h{w}_i = \h{1}.$ 
There is a permutation $\phi:\hel{k}^b \to \el{k}^b$
such that for every $\h{w}\in\hel{k}^b$ and for every $i\in\el{b}$
if $\h{w}_i = \h{1}$ then $\phi\of{\h{w}}_i = \alpha\of{\h{w}, i}.$
\end{lemma}

Now we can describe the reduction.

\begin{lemma} \label{lem:family-to-matchings}
For a function $f:\feln{n}{k}$ there is a full bipartite graph
$G=\p{V_1 \cup V_2, E, w}$ such that
\begin{itemize}
  \item for every $x\in X_f$ there exists a perfect matching $M$ of $G$
    such that $w(M) = x,$
  \item for every perfect matching $M$ of $G$ we have $w(M) \in X_f,$
  \item $\pow{V_1} = \pow{V_2} = O\of{\frac{nk^2\log{k}}{\log{n} + \log{k}}}.$
\end{itemize}
\end{lemma}

\newcommand{\red}[1]{{\color{red}#1}}
\newcommand{\green}[1]{{\color{green}#1}}
\newcommand{\blue}[1]{{\color{blue}#1}}
\newcommand{\cyan}[1]{{\color{cyan}#1}}

\begin{figure}[t]
\begin{center}
  \begin{tabular}{c|c|c|c|c|}
    & 
    $\bt{\h{1}_{\red{(1)}}, \h{1}_{\blue{(2)}}}$ &
    $\bt{\h{1}_{\green{(3)}}, \h{2}}$ &
    $\bt{\h{2}, \h{1}_{\cyan{(4)}}}$ &
    $\bt{\h{2}, \h{2}}$ \\
  \hline
    $\bt{1_{(i)}, 1_{(ii)}}$ &
    $f\of{\red{1}, 1_{(i)}) + f(\blue{2}, 1_{(ii)}}$ &
    $f\of{\green{3}, 1_{(i)}}$ &
    $f\of{\cyan{4}, 1_{(ii)}}$ &
    $0$ \\
  \hline
    $\bt{1_{(i)}, 2_{(ii)}}$ &
    $f\of{\red{1}, 1_{(i)}) + f(\blue{2}, 2_{(ii)}}$ &
    $f\of{\green{3}, 1_{(i)}}$ &
    $f\of{\cyan{4}, 2_{(ii)}}$ &
    $0$ \\
  \hline
    $\bt{2_{(i)}, 1_{(ii)}}$ &
    $f\of{\red{1}, 2_{(i)}) + f(\blue{2}, 1_{(ii)}}$ &
    $f\of{\green{3}, 2_{(i)}}$ &
    $f\of{\cyan{4}, 1_{(ii)}}$ &
    $0$ \\
  \hline
    $\bt{2_{(i)}, 2_{(ii)}}$ &
    $f\of{\red{1}, 2_{(i)}) + f(\blue{2}, 2_{(ii)}}$ &
    $f\of{\green{3}, 2_{(i)}}$ &
    $f\of{\cyan{4}, 2_{(ii)}}$ &
    $0$ \\
  \hline
  \end{tabular}
  \caption{
    The weights on the edges
    of a graph encoding an $f$-family for $f:\feln{4}{2}.$
    The lower indices $(1), (2), (3)$ and $(4)$
    are added to indicate the correspondence between
    the occurrences of $\h{1}$ and the elements of $\el{n}$
    (the first argument of the function $f$).
    The lower indices $(i)$ and $(ii)$ are added to indicate the correspondence
    between the second argument of the function $f$ and the position in the
    (two element) sequence $\langle \cdot, \cdot \rangle.$
  }
  \label{fig:k44-tabular}
\end{center}
\vspace{-1cm}
\end{figure}

\begin{proof}
Let us consider the smallest $b\in\N_+$ such that $c = b \cdot k^{b-1} \geq n.$
Later we will show that $\pow{V_1} = \pow{V_2} = k^b$ is sufficient.

For convenience we extend our chosen function $f:\feln{n}{k}$
to $f:\feln{c}{k}$ in such a way that
for every $i = n + 1, n + 2, \ldots, c$ and for every $j\in\el{k}$
we put $f\of{i, j} = 0.$
Note that the $f$-family $X_f$
does not change after this extension.

Let $V_1 = \hel{k}^b$ and $V_2 = \el{k}^b$ be the sets of words of length
$b$ over the alphabets respectively $\hel{k}$ and $\el{k}.$
Note that $\pow{V_1} = \pow{V_2} = k^b.$

Let $\beta:V_1\times\el{b}\to\el{c}\cup\set{\bot}$
be any function such that
if $\h{w}_j \neq \h{1}$ then $\beta\of{\h{w}, j} = \bot$
and every value from the set $\el{c}$ is used exactly once,
i.e.,
for every $x\in\el{c}$ there is exactly one argument
$(\h{w}, j) \in V_1 \times \el{b}$
such that
$\beta\of{\h{w}, j} = x.$
Note that such a function always exists because the total number of the
occurrences of $\h{1}$ in all the words in $V_1$
is exactly $c = b \cdot k^{b-1}.$

Now we define our weight function $w:V_1\times V_2\to\N$
as follows
\[
  w\of{\h{t}, u} =
    \sum_{\substack{i\in\el{b} \\ \beta\of{\h{t}, i}\neq\bot}}
      f\of{\beta\of{\h{t}, i}, u_i}.
\]
An example of such weight function can be found in
Figure~\ref{fig:k44-tabular}
(or in the Appendix in Figure~\ref{fig:k44} as a picture of a bipartite
graph).

Note that because $\beta$ picks every value from the set $\el{c}$
exactly once then for every permutation $\phi:V_1\to V_2$ we have
\[
  \sum_{\h{t}\in V_1} w\of{\h{t}, \phi\of{\h{t}}} =
  \sum_{\h{t}\in V_1}
    \sum_{\substack{i\in\el{b} \\ \beta\of{\h{t}, i}\neq\bot}}
      f\of{\beta\of{\h{t}, i}, \phi\of{\h{t}}_i}
  \in X_f.
\]
In other words the set of the weights of all perfect matchings in $G$
is a subset of $X_f$ as required.

We also need to show that for every $x \in X_f$ there exists some
permutation $\phi:V_1\to V_2$ such that
$\sum_{\h{t}\in V_1} w\of{\h{t}, \phi\of{\h{t}}} = x.$
This permutation gives us a corresponding perfect matching of weight $x$
in $G.$ 

Let us take a function $\sigma:\fel{c}{k}$ such that 
$x = \sum_{i\in\el{c}} f\of{i, \sigma\of{i}},$
which exists by the definition of $X_f.$
Define $\alpha:\hel{k}^b\times\el{b}\to\el{k}\cup\set{\bot}$ as follows:
\[
  \alpha\of{\h{u}, i} =
  \begin{cases}
    \sigma\of{\beta\of{\h{u}, i}} & \mbox{for } \h{u}_i = \h{1}\\
    \bot & \mbox{for } \h{u}_i \neq \h{1}.\\
  \end{cases}
\]

Now we can use Lemma~\ref{lem:bijection} with function $\alpha$
to obtain a permutation $\phi:\hel{k}^b \to \el{k}^b$ such that
for every $\h{u}\in\hel{k}^b$ and for every $i\in\el{b}$
if $\h{u}_i = \h{1}$ then $\phi\of{\h{u}}_i = \sigma\of{\beta\of{\h{u}, i}}.$

So we have that
\[
  \sum_{\h{u}\in\hel{k}^b} w\of{\h{u}, \phi\of{\h{u}}} =
  \sum_{\h{u}\in\hel{k}^b}
    \sum_{\substack{i\in\el{b} \\ \beta\of{\h{u}, i}\neq\bot}}
      f\of{\beta\of{\h{u}, i}, \sigma\of{\beta\of{\h{u}, i}}}
 = \sum_{i\in\el{c}}  f\of{i, \sigma\of{i}}
 = x.
\]

Hence we have shown that
$X_f$ is the set of weights of all perfect matchings in graph $G.$
The last thing is to show that the number of the vertices
is sufficiently small.
We know that $\p{b - 1} \cdot k^{b - 2} < n$ so
$\p{b - 1} \cdot k^{b - 1} < nk$
and then
$b - 1 < \log_{k}{\rnlogkn{nk}{k}}.$
Therefore
$k^b
  < k \cdot \rnlogkn{nk}{k}
  = O\of{\frac{nk^2}{\log_k{nk}}}
  = O\of{\frac{nk^2\log{k}}{\log{n} + \log{k}}}.$
So $|V_1| = |V_2| = k^b = O\of{\frac{nk^2\log{k}}{\log{n} + \log{k}}},$
as required.
\end{proof}

Lemma~\ref{lem:family-to-matchings} immediately implies the following result.

\begin{lemma} \label{lem:FInt-to-CMW}
There is a polynomial time reduction that
for an instance $I = (f, g)$ of \FInt with
$f:\feln{a}{b}$ and $g:\feln{c}{d}$ reduces it into
an instance of \CMW $J = (G_1, G_2)$ with
$\pow{V\of{G_1}} = O\of{\frac{a b^2 \log{b}}{\log{a} + \log{b}}}$
and $\pow{V\of{G_2}} = O\of{\frac{c d^2 \log{d}}{\log{c} + \log{d}}}$
vertices.
The sets of the weights of all perfect matchings in $G_1$ and in $G_2$
are equal respectively to $X_f$ and $X_g.$
\end{lemma}

Together with Lemma~\ref{lem:ThreeSAT-to-FInt} we obtain the following
theorem.

\begin{theorem} \label{thm:ThreeSAT-to-CMW}
There is a polynomial time reduction from a given instance of \ThreeSAT with
$n$ variables and $m$ clauses
into an instance of \CMW with $|V(G_1)| = O\of{\ndivlog{n}},$
$|V(G_2)| = O\of{\ndivlog{m}}$
and the maximum matching weights bounded by $2^{3m}.$
\end{theorem}

A commonly know corollary of the Sparsification Lemma from \cite{PZ01} is:

\begin{corollary} \label{cor:ThreeSAT-hardness}
There is no algorithm solving \ThreeSAT in $2^{o(n + m)}$-time
where $n$ is the number of variables
and $m$ is the number of clauses
unless ETH fails.
\end{corollary}

Using Corollary~\ref{cor:ThreeSAT-hardness}
we can prove the following lower bound.

\begin{corollary} \label{cor:CMW-hardness}
There is no algorithm solving \CMW in $2^{o(n\log n)}\poly{r}$-time
where $n$ is the total number of vertices,
and $r$ is the bit size of the input,
unless ETH fails.
\end{corollary}

\begin{proof}
For a given instance of \ThreeSAT with $n$ variables and $m$ clauses we can use
the reduction from Theorem~\ref{thm:ThreeSAT-to-CMW} to obtain an instance of
\CA.
The total number of the vertices in the new instance is
$$|V(G_1)| + |V(G_2)| =  O\of{\ndivlog{n} + \ndivlog{m}}
  = O\of{\frac{n + m}{\log\of{n + m}} + \frac{n + m}{\log\of{n + m}}}
  = O\of{\frac{n + m}{\log\of{n + m}}}$$
because the function $\ndivlog{n}$ is nondecreasing for the sufficiently big
values of $n.$
Weights of the matchings are bounded by $2^{3m}$ and the bit size
of the instance
$$r = O\of{\p{\frac{n + m}{\log\of{n + m}}}^2 \log 2^{3m}} = \poly{nm}.$$
Then let us assume that there is an algorithm solving \CMW in
$2^{o(n\log n)} \poly{r}$-time.
Then we could solve our instance in time
$$2^{o\of{\frac{n + m}{\log(n + m)}
    \log\of{\frac{n + m}{\log(n + m)}}}} \poly{\poly{nm}}
= 2^{o\of{\frac{n + m}{\log(n + m)}
    \log(n + m)}} \poly{nm}
= 2^{o(n + m)}
$$
which contradicts ETH by Corollary~\ref{cor:ThreeSAT-hardness}.
\end{proof}

\section{Hardness of \CA}
\label{sec:CA}

Consider two weighted
full bipartite graphs $G_1$ and $G_2.$ We would like to encode them
in a \CA instance in such a way that this \CA instance is a YES-instance
if and only if there are two perfect matchings,
one in $G_1$ and the other in $G_2,$ of the same weight. 

Consider an instance $I=(V,d,s)$ of \CA.
We say that $c:V\to\Z$ is a YES-coloring if $c$ is a proper coloring and has
span at most $s.$
Note that an instance of \CA is a YES-instance
if and only if it has a YES-coloring.

Our approach is that we encode those graphs $G_1$ and $G_2$
separately in such a way
that we have a special vertex $v_M$ whose color in every YES-coloring
represents a weight of some perfect matching in $G_1$
and on the other hand in every YES-coloring
its color represents (in a similar way) a weight
of some perfect matching in $G_2.$
So a YES-coloring
coloring would be possible if and only if the graphs $G_1$ and $G_2$
have two perfect matchings, one in $G_1$ and the other in $G_2$, with equal
weights.

Before we present a way to encode a weighted full bipartite graph in a
\CA instance we would like to present the two lemmas to merge those two encoded
graphs into a one instance of \CA.
In order to do that we use the following concepts.

We say that instance $I$ is \emph{$(x, y)$-spanned}
for some vertices $x, y\in V$
if for every YES-coloring $c$ of $I$
we have $\abs{c\of{x} - c\of{y}} = s - 1.$

We say that an instance $I = (V, d, s)$ of \CA is \emph{$(X, Y)$-spanned}
for some nonempty subsets of the vertices $\emptyset\neq X, Y\sseq V$
if it is $(x, y)$-spanned for every two vertices $x\in X$ and $y\in Y.$

\begin{lemma} \label{lem:ca-merge} ($\bigstar$)
For every $(u, v)$-spanned instance $I_1=(V_1, d_1, s)$ and
$(w, z)$-spanned instance $I_2=(V_2, d_2, s)$
of \CA
there is a $(\set{u,w}, \set{v,z})$-spanned instance $I=(V_1 \cup V_2, d, s)$
of \CA
such that
\begin{enumerate}[(i)]
  \item for every YES-coloring $c$ of $I$
    the coloring $c\mid_{V_1}$ is a YES-coloring of $I_1$
    and the coloring $c\mid_{V_2}$ is a YES-coloring of $I_2,$
  \item for every YES-coloring $c_1$ of $I_1$
    and every YES-coloring $c_2$ of $I_2$
    such that $c_1\of{u} = c_2\of{w},$ $c_1\of{v} = c_2\of{z}$ and
    for every $x\in V_1 \cap V_2$ we have $c_1\of{x} = c_2\of{x}$
    there exists a YES-coloring $c$ of $I$
    such that $c\mid_{V_1} = c_1$ and $c\mid_{V_2} = c_2.$
\end{enumerate}
\end{lemma}

\begin{lemma} \label{lem:ca-extend} ($\bigstar$)
For every $(v_L, v_R)$-spanned instance $I=(V, d, s)$ of \CA and
for every numbers $l, r \in \N$ there exists
a $(w_L, w_R)$-spanned instance $I'=(V \cup \set{w_L, w_R}, d', l + s + r)$
such that
\begin{enumerate}[(i)]
  \item for every YES-coloring $c$ of $I$
    there is a YES-coloring $c'$ of $I'$ such that $c'\mid_V = c,$
  \item for every YES-coloring $c'$ of $I'$ such that
    $c'\of{w_L} \leq c'\of{w_R}$ we have that 
    \begin{itemize}
      \item a coloring $c'\mid_V$ is a YES-coloring of $I,$
      \item $c'\of{v_L} = c'\of{w_L} + l$ and
        $c'\of{v_R} = c'\of{w_R} - r.$
    \end{itemize}
\end{enumerate}
\end{lemma}

The proofs of these two lemmas are straightforward.
They are moved to the Appendix due to space limitations.

\begin{lemma} \label{lem:matchings-to-ca}
Let $G = (V_1 \cup V_2, E, w)$ be a weighted full bipartite graph
with nonnegative weights and such that $\pow{V_1} = \pow{V_2}.$
Let $n = \pow{V_1},$
$m = \max_{e\in E} w(e),$
$M = n \cdot m + 1,$
$l = \p{4n - 1} \cdot M$ and
$s = (8n - 1) \cdot M.$
There exists a $(v_L, v_R)$-spanned instance $I=(V, d, s)$ of \CA
with $\pow{V} = O\of{n}$ and
such that for some vertex $v_M \in V,$
\begin{enumerate}[(i)]
  \item for every YES-coloring $c$ of $I$
    such that $c\of{v_L} \leq c\of{v_R}$
    there exists a perfect matching $M_G$ in $G$ such that
    $c\of{v_M} = c\of{v_L} + l + w\of{M_G}$ and
  \item for every perfect matching $M_G$ in $G$ there exists a YES-coloring
    $c$ of $I$ such that $c\of{v_L} \leq \of{v_R}$ and
    $c\of{v_M} = c\of{v_L} + l + w\of{M_G}.$
\end{enumerate}
\end{lemma}

\begin{figure}[t]
\begin{center}
  \def\svgwidth{5in}
  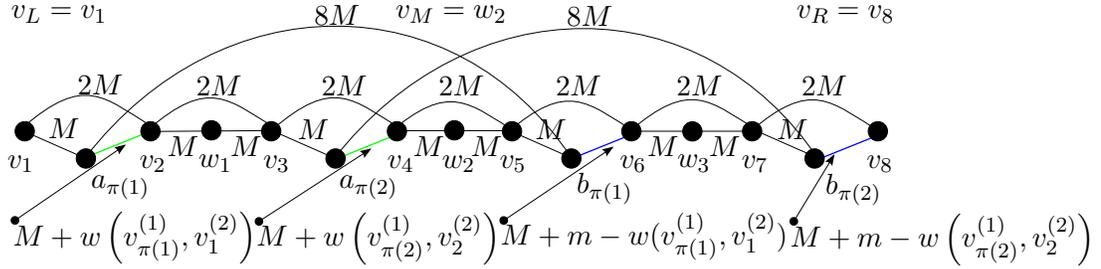
  \vspace{0.2cm}
  \caption{
    A weighted full bipartite graph $(G, w)$ with $|V_1| = |V_2| = 2$
    encoded in a \CA form.
    The color of the vertex $v_M = w_2$ corresponds to the weight of the
    perfect matching in $G$ given by the permutation $\pi$
    and is equal to $c\of{v_M} = c\of{v_L} + 7M + w\of{M_{\pi}}.$
    The picture is simplified.
    Some of the edges and corresponding to them minimum distances
    are omitted in the picture.
  }
  \label{fig:matching-to-ca}
\end{center}
\vspace{-1cm}
\end{figure}

\newcommand{\vind}[3]{{#1}^{\p{#2}}_{#3}}
\newcommand{\vv}[2]{\vind{v}{#1}{#2}}
\newcommand{\ww}[2]{\vind{w}{#1}{#2}}

\newcommand{\maxv}{4n}
\newcommand{\seqv}{\seq{v}{\maxv}}
\newcommand{\elv}{\el{\maxv}}

\newcommand{\maxw}{2n - 1}
\newcommand{\seqw}{\seq{w}{\maxw}}
\newcommand{\elw}{\el{\maxw}}

\newcommand{\maxa}{n}
\newcommand{\seqa}{\seq{a}{\maxa}}
\newcommand{\ela}{\el{\maxa}}

\newcommand{\maxb}{n}
\newcommand{\seqb}{\seq{b}{\maxb}}
\newcommand{\elb}{\el{\maxa}}

\begin{proof}
Let $V_1 = \set{\vv{1}{1}, \vv{1}{2}, \ldots, \vv{1}{n}}$
and $V_2 = \set{\vv{2}{1}, \vv{2}{2}, \ldots, \vv{2}{n}}.$
We will build our \CA instance step by step.
A simplified picture of
the instance can be found in Figure~\ref{fig:matching-to-ca}.

Let us introduce the vertices $v_L = \seqv = v_R$ to set $V.$
Because of the symmetry we can assume that for every coloring $c$ of our
instance we have $c\of{v_L} \leq c\of{v_R}.$

We set the minimum distance $d\of{v_L, v_R} = s - 1.$
Then for every YES-coloring $c$ we have that
$\abs{c\of{v_L} - c\of{v_R}} = s - 1$ so our instance is $(v_L, v_R)$-spanned.

For every $i, j \in \elv$ such that $i \neq j$ and
$\set{i, j} \neq \set{1, \maxv}$
we set the minimum distance
$d\of{v_i, v_j} = \abs{i - j} \cdot 2M.$
Then we can prove the following claim.

{\bf Claim 1}
For every YES-coloring $c$ and for every
$i < j$ we have that $c\of{v_i} < c\of{v_j}.$

{\em Proof of the claim:}
Indeed if the colors $c\of{v_1}, c\of{v_2}, \ldots, c\of{v_{\maxv}}$
are not strictly increasing or strictly decreasing then the distances
between the consecutive colors of the set $c\of{\set{\seqv}}$
are at least $2M$ and at least one of them is at least $4M.$
So it would need to use at least
$\p{\maxv - 2} \cdot 2M + 4M + 1 > (8n - 1) \cdot M = s$
colors.
In addition we have assumed that for all colorings
$c\of{v_L} \leq c\of{v_R}$ so the colors are strictly increasing.
This proves the claim.

Note that for every YES-coloring and for every $i \in \el{\maxv - 1}$
we have
\begin{equation} \label{eq:vdists}
  2M \leq c\of{v_{i + 1}} - c\of{v_i} \leq 2M + n\cdot m < 3M,
\end{equation}
for otherwise $c(v_R) - c(v_L) \ge (4n - 2) \cdot 2M + 2M + n \cdot m + 1 = s,$
so $c$ has the $\myspan$ at least $s + 1,$ a contradiction.

Let us introduce new vertices $\seqw$ to set $V.$
For every $i\in\elw$ and $j\in\elv$
we set the minimum distances
$d\of{w_i, v_j} = \abs{4i + 1 - 2j} \cdot M.$
For every YES-coloring $c$ and for every $i\in\elw$
we have
$c\of{v_{2i}} + M \leq c\of{w_i} \leq c\of{v_{2i + 1}} - M$
by \eqref{eq:vdists},
for otherwise we have that $c(v_j)\leq c(w_i) \leq c(v_{j + 1})$
for some $j \neq 2i$
(because $c(v_{\maxv}) - c(v_1) = s - 1$ so every YES-coloring uses only
colors from the interval $[c(v_1), c(v_{\maxv})]$)
and then
$c(v_{j + 1}) - c(v_j) \ge d(v_j, w_i) + d(w_i, v_{j + 1})$
and $\{v_j, v_{j + 1}\} \ne \{v_{2i}, v_{2i + 1}\}$
so at least one of these two distances is at least $3M$
and therefore $c(v_{j + 1}) - c(v_j) \ge 3M + M > 3M,$ a contradiction
with \eqref{eq:vdists}.
Thus infer the following claim.

{\bf Claim 2} For every YES-coloring $c$ the colors of the vertices in the
sequence
\begin{equation} \label{eq:vwseq}
  v_1, v_2, w_1, v_3, v_4, w_2, v_5\ldots,
    v_{\maxv - 2}, w_{\maxw}, v_{\maxv - 1}, v_{\maxv}
\end{equation}
are increasing.

We introduce new vertices $\seqa$ and
for every $i\in\ela$ and $j\in\elv$ we set the minimum distances
$$d\of{a_i, v_j} = \begin{cases}
  M + w\of{\vv{1}{i}, \vv{2}{j/2}} &
    \mbox{when } j \leq 2n \mbox{ and } 2 \mid j\\
  M & \mbox{when } j \leq 2n \mbox{ and } 2 \nmid j\\
  \p{j - 2n} \cdot 2M + M & \mbox{when } j > 2n.\\
\end{cases}$$
Then for every YES-coloring $c$ and for every $i\in\ela$ we have
$c\of{a_i} \leq c\of{v_{2n}}$
because in other case we have $c(v_j) \leq c(a_i) \leq c(v_{j + 1})$
for some $j \ge 2n$ and then
$c(v_{j + 1}) - c(v_j) \ge d(v_j, a_i) + d(a_i, v_{j + 1}) \ge M + 3M > 3M,$
a contradiction with \eqref{eq:vdists}.

Moreover for every $i\in\ela$ and every $j\in\elw$ we set the minimum distance
$d\of{a_i, w_j} = 2M.$
Therefore by \eqref{eq:vdists} and \eqref{eq:vwseq}
for every YES-coloring $c$ and every $i\in\ela$ the vertex $a_i$
is colored with the color from one of the intervals
$\p{c\of{v_{2j - 1}}, c\of{v_{2j}}}$ for some $j \in \el{n}.$

Finally for every $i, j \in \ela$ such that $i \neq j$
we set the minimum distance
$d\of{a_i, a_j} = 4M$ so by \eqref{eq:vdists} we know that
for every YES-coloring $c$ and every $i\in \el{n}$
exactly one one vertex $a_j$ of the vertices $\seqa$ is colored with the
color from the interval $\p{c\of{v_{2i - 1}}, c\of{v_{2i}}}.$
The assignment of vertices $\seqa$ to intervals
$(c(v_1), c(v_2)), (c(v_3), c(v_4)), \ldots, (c(v_{2n - 1}), c(v_{2n}))$
determines a permutation $\pi_c:[n]\to[n],$ i.e., $\pi_c(i) = j$
if $a_j$ gets a color from $(c(v_{2i - 1}), c(v_{2i})).$
Hence we get the following claim:

{\bf Claim 3} For every YES-coloring $c$ there is a permutation $\pi_c$ such
that the colors of the vertices of the sequence
$$v_1, a_{\pi_c\of{1}}, v_2, w_1, v_3, a_{\pi_c\of{2}}, v_4, w_2, v_5 \ldots,
  v_{2n - 1}, a_{\pi_c\of{n}}, v_{2n}$$
are increasing.

Similarly we introduce new vertices $\seqb$ and
for every $i\in\elb$ and $j\in\elv$ we set the minimum distances
$$d\of{b_i, v_j} = \begin{cases}
  \p{2n - j + 1} \cdot 2M + M & \mbox{when } j \leq 2n \\
  M + m - w\of{\vv{1}{i}, \vv{2}{j/2 - n}} &
    \mbox{when } j > 2n \mbox{ and } 2 \mid j \\
  M & \mbox{when  } j > 2n \mbox{ and } 2 \nmid j.\\
\end{cases}$$
Also for every $i\in\elb$ and every $j\in\elw$ we set the minimum distance
$d\of{b_i, w_j} = 2M$ and for every $i, j \in \elb$ such that $i \neq j$
we set the minimum distance $d\of{b_i, b_j} = 4M.$
Hence similarly as before,
for every YES-coloring $c$ and every $i\in\el{n}$ exactly one
vertex $b_j$ of the vertices $\seqb$ is colored with the color from
the interval $\p{c\of{v_{2n + 2i - 1}}, c\of{v_{2n + 2i}}}.$
Analogously as before, the colors of the vertices $\seqb$
determine a permutation $\rho_c:[n]\to[n].$
Thus we have the following claim.

{\bf Claim 4} For every YES-coloring $c$ there is a permutation $\rho_c$
such that the colors of the vertices in the sequence
$$v_{2n + 1}, b_{\rho_c\of{1}}, v_{2n + 2}, w_{n + 1},
  v_{2n + 3}, b_{\rho_c\of{2}}, v_{2n + 4}, w_{n + 2}, v_{2n + 5}\ldots
  v_{4n - 1}, b_{\rho_c\of{n}}, v_{4n}$$
are increasing.

For every $i\in\el{n}$ we set the minimum distance
$d\of{a_i, b_i} = n \cdot 4M.$
Then we know that for every YES-coloring $c$ we have
$\pi_c^{-1}\of{i} \leq \rho_c^{-1}\of{i}$
for otherwise we can take
$j = 2\pi_c^{-1}(i) - 1$ and
$k = 2n + 2\rho_c^{-1}(i)$
and then
$(c(b_i) - c(a_i)) + 2M \le c(v_j) - c(v_k)$
and $k - j \le 2n$
so the sequence $v_1, v_2, \ldots, v_j, v_k, \ldots, v_{\maxv}$
has at least $\maxv - 2n + 1 = 2n + 1$ elements
so
$ c(v_{\maxv}) - c(v_1)
\ge (2n - 1) \cdot 2M + (c(v_k) - c(v_j))
\ge (2n - 1) \cdot 2M + (c(b_i) - c(a_i)) + 2M
\ge (2n - 1) \cdot 2M + n \cdot 4M + 2M
= n \cdot 8M
> (n - 1) \cdot 8M  - 1
= s - 1,$
a contradiction.
Since $\pi_c$ and $\rho_c$ are permutations, we further infer that for every
YES-coloring $c$ we have $\pi_c = \rho_c.$
Hence we have the following claim.

{\bf Claim 5} For every YES-coloring $c$ there is a permutation $\pi_c$
such that the colors of the vertices in the sequence
$$v_1, a_{\pi_c\of{1}}, v_2, w_1, v_3, a_{\pi_c\of{2}}, v_4, w_2, v_5 \ldots,
  v_{2n - 1}, a_{\pi_c\of{n}}, v_{2n}, w_n,$$ $$
  v_{2n + 1}, b_{\pi_c\of{1}}, v_{2n + 2}, w_{n + 1},
  v_{2n + 3}, b_{\pi_c\of{2}}, v_{2n + 4}, w_{n + 2}, v_{2n + 5}\ldots
  v_{4n - 1}, b_{\pi_c\of{n}}, v_{4n}$$
are increasing.

This ends the description of the instance $I.$
Note that $I$ is $(v_L, v_R)$-spanned because $d(v_L, v_R) = s - 1.$
Let us put $v_M = w_n.$
We are going to show the following claim.

{\bf Claim 6} ($\bigstar$)
Let $\pi:\el{n}\to\el{n}$ be any permutation and
$M_{\pi} = \set{\vv{1}{i} \vv{2}{\pi\of{i}} : i\in\el{n}}$ be 
the corresponding perfect matching in $G.$
There is exactly one YES-coloring $c$ such that
$\pi_c = \pi.$
Moreover $c(v_M) = c(v_L) + l + w(M_{\pi}).$

To proove the claim it is sufficient to check all the introduced
minimum allowed distances $d$ and the $\myspan$
for the coloring implied by the sequence as in Claim 5.
This is a simple manual check but due to its length
the proof of the claim is moved to the Appendix.

Thus there is a one-to-one correspondence between permutations and
YES-colorings.
Moreover we know that for every YES-coloring $c$ we have
$c\of{v_M} = c\of{v_L} + l + w\of{M_{\pi_c}}$
where $M_{\pi_c}$ is the perfect matching in $G$ corresponding
to permutation $\pi_c.$
Hence we have shown (i) and (ii) as required.
\end{proof}

\begin{lemma} \label{lem:CMW-to-CA} ($\bigstar$)
There is a polynomial time reduction such that
for a given instance $I=(G_1, G_2)$ of \CMW
with $n_1 = \pow{V\of{G_1}}, n_2 = \pow{V\of{G_2}}$
and such that the weight functions of $G_1$ and $G_2$
are bounded by respectively $m_1$ and $m_2$
reduces it into an instance of \CA with
$O\of{n_1 + n_2}$ vertices
and the maximum edge weight in $O\of{n_1^2 m_1 + n_1^2 m_2}.$
\end{lemma}

\begin{figure}[t]\begin{center}
  \def\svgwidth{5in}
  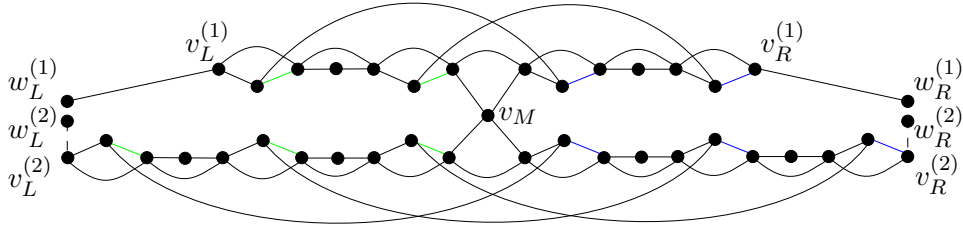
  \caption{
    Two weighted full bipartite graphs $(G_1, w_1)$ (with $n_1 = 2$)
    and $(G_2, w_2)$ (with $n_2 = 3$)
    encoded in a \CA form.
    The color of the vertex $v_M = w_2$ corresponds to the weight of some
    perfect matching in $G_1$
    and to the weight of some perfect matching in $G_2.$
    These two weights have to be equal.
    The picture is simplified.
    Some of the edges are omitted in the picture.
    Note that the values $M$ and $m$ can be different for $(G_1, w_1)$ and
    for $(G_2, w_2)$.
  }
  \label{fig:matchings-to-ca}
\end{center}
\vspace{-1cm}
\end{figure}

In the proof we use Lemma~\ref{lem:matchings-to-ca} to encode
$G_1$ and $G_2$ in two instances of \CA,
then we extend them to the common length using
Lemma~\ref{lem:ca-extend}
and finally we merge them using Lemma~\ref{lem:ca-merge}.
The proof is straightforward and is moved to the Appendix
due to space limitations.
The simplified picture of the obtained \CA instance can
be found in Figure~\ref{fig:matchings-to-ca}.

Now we can use the results of Section~\ref{sec:CMW} to get following
two corollaries.

\begin{corollary}
There is no algorithm solving \CA in
$2^{o(n\log n)} \poly{r}$
where $n$ is the number of the vertices
and $r$ is the bit size of the instance
unless ETH fails.
\end{corollary}

\begin{proof}
For a given instance of \CMW with $n$ vertices and the weights bounded by $m$
we can transform it by
Lemma~\ref{lem:CMW-to-CA} into an instance of \CA with $n' = O(n)$ vertices
and the weights bounded by $\ell = O(n^2m).$
Note that for the bit size $r'$ of the new instance we have
$\poly{r'} = \poly{(n')^2\ell} = \poly{n, m} = \poly{r}.$

Let us assume that we can solve \CA in $2^{o(n\log n)}\poly{r}$-time.
Then we can solve our instance in time
$2^{o(n'\log n')}\poly{r'} = 2^{o(n\log n)} \poly{r}$
which contradicts ETH by Corollary~\ref{cor:CMW-hardness}.
\end{proof}

\begin{corollary} \label{cor:ca-nloglogl} ($\bigstar$)
There is no algorithm solving \CA in
$2^{n \cdot o(\log\log\ell)} \poly{r}$
where $n$ is the number of the vertices
and $r$ is the bit size of the instance
unless ETH fails.
\end{corollary}

\subsection*{Acknowledgments}

We thank {\L}ukasz Kowalik who was the originator of this studies for
a comprehensive support.

\bibliographystyle{abbrv}
\bibliography{noexp}

\newpage
\appendix

\section{Omitted Proofs from Section~\ref{sec:CMW}}

\subsection{Omitted Proofs from Section~\ref{subsec:ThreeSAT-to-FInt}}

\begin{lemma}[Lemma~\ref{lem:ThreeSAT-to-FInt}]
There is a polynomial time reduction from a given instance of \ThreeSAT with
$n$ variables and $m$ clauses
into an instance of \FInt with $f:\feln{n}{2}$ and $g:\feln{m}{7}$
such that $\max{X_f} < 2^{3m}$ and  $\max{X_g} < 2^{3m}.$
\end{lemma}

\begin{proof}
Let $V = \set{\seq{v}{n}}$
and $C = \set{\seq{c}{m}}$ be
the sets of variables and clauses of the input formula, respectively.

Let $D = \set{\seq{d}{3m}}$ be the set of
all $3m$ occurrences of our $n$ variables in our $m$ clauses.
We will treat these occurrences as separate variables.

For every variable $v_i\in V$ we define a set $I_i\sseq\el{3m}$
such that $j\in I_i$ if and only if $d_j$ is an occurrence
of the variable $v_i.$

Similarly for every clause $c_i\in C$ we define a set $J_i\sseq\el{3m}$
such that $j\in J_i$ if and only if $d_j$ an occurrence (of any variable)
belonging to the clause $c_i.$
For every $i\in\el{m}$ we have $\abs{J_i} = 3.$

For every clause $c_i$ we can treat the subsets of $J_i$ as
the assignments of the occurrences $d_j$ belonging to the clause $c_i.$
We treat the subset $K\sseq J_i$ as the assignment of the occurrences in the
clause $c_i$ such that the occurrence $d_j$ is set to $1$ if and only if
$j \in K,$ otherwise it is set to $0.$
We say that $K\sseq J_i$ satisfies the clause $c_i$ if the corresponding
assignment of the occurrences satisfies this clause.

For every clause $c_i\in C$ let us define the set
$P_i = \set{K\sseq J_i : \mbox{$K$ satisfies the clause $c_i$}}.$
Again note that we treat here all the occurrences as the different variables.
Note that $\abs{P_i} = 7$ for every $i,$
so we can denote $P_i = \set{P_i^1, P_i^2, \ldots, P_i^7}.$

A number from $0, 1, \ldots, 2^{3m} - 1$ can be interpreted in the binary
system as the characteristic vector of length $3m$
of a subset of the indices of the occurrences
i.e., that the $i$-th bit represents if the
occurrence $d_i$ belongs to this subset or not.

We define a function $f:\feln{n}{2}$ such that for every $i\in\el{n}$
we set $f\of{i, 1} = \sum_{j \in I_i} 2^{j - 1}$
and $f\of{i, 2} = 0.$
In other words the number $f\of{i, 1}$ represents the characteristic vector of
all the occurrences of the variable $v_i.$

Note that $\{\sigma: [n]\to[2]\}$ corresponds to the set of all assignments
of variables.
Therefore $X_f$ is the set of all the characteristic vectors
which represent all the assignments of the occurrences such that all the
occurrences of the same variable have the same value.

We define a function $g:\feln{m}{7}$ such that for every $i\in\el{m}$
and for every $j\in\el{7}$ we can set
$g\of{i, j} = \sum_{k\in P_i^j} 2^{k - 1}.$
Then for every $i\in\el{m}$ the numbers
$g\of{i, 1}, g\of{i, 2}\ldots, g\of{i, 7}$
represent the characteristic vectors of all the assignments of the occurrences
in the clause $c_i$ which satisfy this clause.

Therefore the set $X_g$ is the set of all the characteristic vectors
which represents the assignments of all $3m$ occurrences such that all the
clauses are satisfied.

It follows that
the set $X_f \cap X_g$ is the set of all the characteristic vectors
which represent the assignments of the occurrences such that all the
occurrences of the same variable have the same value
and all the clauses are satisfied.
In other words, elements of $X_f \cap X_g$ correspond to satisfying
assignments.
\end{proof}

\subsection{Omitted Proofs from Section~\ref{subsec:FInt-to-CMW}}

The first lemma provides a way of merging $k$ permutations 
$\phi_1, \phi_2, \ldots, \phi_k: \hel{k}^b \to \el{k}^b$
into one permutation $\phi: \hel{k}^{b + 1} \to \el{k}^{b + 1}$
in a way specified by a function $\rho: \hel{k}^b \to \el{k}.$

\begin{lemma} \label{lem:bijection-merge}
For every $b\in \N$ and
for a given sequence of permutations
$\phi_1, \phi_2, \ldots, \phi_k: \hel{k}^b \to \el{k}^b$
and for every function $\rho: \hel{k}^b \to \el{k}$
there is a permutation $\phi: \hel{k}^{b + 1} \to \el{k}^{b + 1}$
such that 
\begin{enumerate}[(i)]
  \item for every $x\in\el{k}$ and for every $\h{w}\in\hel{k}^b$
    there exists $y\in\el{k}$ such that
    $\phi\of{\h{x}\h{w}} = y \phi_x\of{\h{w}}$ and moreover
  \item for every $\h{w}\in\hel{k}^b$ we have
    $\phi\of{\h{1}\h{w}} = \rho\of{\h{w}} \phi_1\of{\h{w}}.$
\end{enumerate}
\end{lemma}

Before we proceed to the proof we suggest the reader to take a look at an
example in Figure~\ref{fig:matching-merge} ($b=1,$ $k=3$).

\begin{figure}[t]
\begin{center}
  \def\svgwidth{5in}
  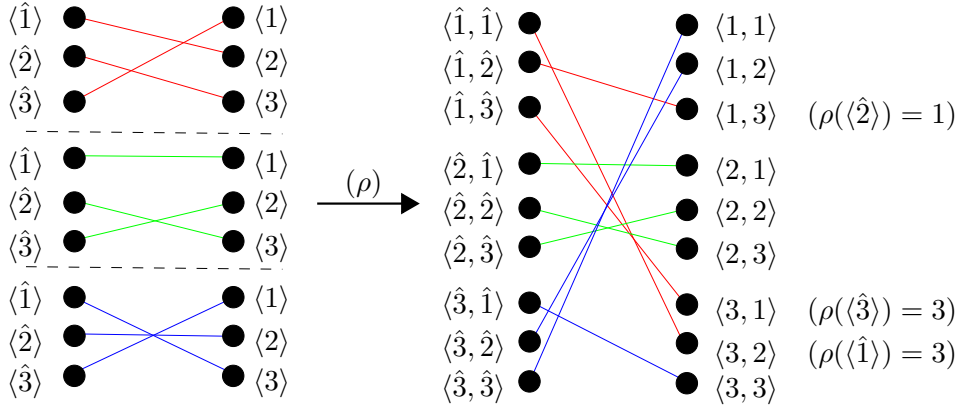
  \caption{
    Merging three permutations (presented as perfect matchings)
    with respect to the function $\rho$ such that
    $\rho\of{\langle\h{1}\rangle} = 3,$ $\rho\of{\langle\h{2}\rangle} = 1$
    and $\rho\of{\langle\h{3}\rangle} = 3.$
  }
  \label{fig:matching-merge}
\end{center}
\vspace{-0.5cm}
\end{figure}

\begin{proof}
We start with a permutation
$\h{x}\h{w} \mapsto x \phi_x\of{\h{w}}$
which already satisfies the condition
$\exists_y \phi\of{\h{x}\h{w}} = y \phi_x\of{\h{w}}.$
Then we are going to swap the values for some (disjoint) pairs of the arguments
in order to fulfill the condition
$\phi\of{\h{1}\h{w}} = \rho\of{\h{w}} \phi_1\of{\h{w}}.$
Such swaps are preserving the condition of being permutation.
Moreover we perform only such swaps that preserve also the
$\exists_y \phi\of{\h{x}\h{w}} = y \phi_x\of{\h{w}}$ condition.

For every $\h{w}\in\hel{k}^b$ we need to put
$\phi\of{\h{1}\h{w}} = \rho\of{\h{w}} \phi_1\of{\h{w}}.$
Let us assign $x = \rho{\of{\h{w}}}$ and
$\h{u} = \phi_x^{-1}\of{\phi_1\of{\h{w}}}.$
Note that $\phi_x\of{\h{u}} = \phi_1\of{\h{w}}.$
If $x \neq 1$ then to avoid a collision
$\phi\of{\h{1}\h{w}} = \phi\of{\h{x}\h{u}}$
we can put $\phi\of{\h{x}\h{u}} = 1\phi_1\of{\h{w}}.$
So we have swapped the values for the arguments $\h{1}\h{w}$ and $\h{x}\h{u}.$
Our function is still a permutation.
Note that the condition $\exists_y \phi\of{\h{x}{u}} = y\phi_x{\h{u}}$
is still preserved because $\phi_x\of{\h{u}} = \phi_1\of{\h{w}}.$
We just need to show that the swaps can be performed independently.

For every $i\in\el{k}$ a function $\phi_i^{-1} \circ \phi_1$ is a permutation
so for every $\h{w}\in\hel{k}^b$ the values of
$\rho\of{\h{w}} \phi_{\rho\of{\h{w}}}^{-1}\of{\phi_1\of{\h{w}}}$
are pairwise different.
Indeed for two different $\h{u}, \h{w} \in\hel{k}^b$
either the values $\rho\of{\h{u}}$ and $\rho\of{\h{w}}$ are different
or $\rho\of{\h{u}} = \rho\of{\h{w}} = x$ for some $x\in\el{k}$ and then
$\p{\phi_x^{-1} \circ \phi_1}\of{\h{u}} \neq
  \p{\phi_x^{-1} \circ \phi_1}\of{\h{w}}$
so then the values
$\phi_{\rho\of{\h{u}}}^{-1}\of{\phi_1\of{\h{u}}}$ and
$\phi_{\rho\of{\h{w}}}^{-1}\of{\phi_1\of{\h{w}}}$ are different.
Therefore our pairs of the arguments to swap are pairwise disjoint.
Thus all the swaps can be performed independently.

So for every $x\in\el{k}$ and $\h{w}\in\hel{k}^b$ we have
$$\phi\of{\h{x}\h{w}} = \begin{cases}
  \rho\of{\h{w}} \phi_1\of{\h{w}} & \mbox{for } \h{x} = \h{1} \\
  1 \phi_x\of{\h{w}} & \mbox{for }
    \h{x} \neq \h{1} \wedge \rho\of{\phi_1^{-1}\of{\phi_x\of{\h{w}}}} = x\\
  x \phi_x\of{\h{w}} & \mbox{in other cases.}\\
\end{cases}$$
\end{proof}

\begin{lemma}[Lemma~\ref{lem:bijection}]
Let $b\in \N$ and $\alpha:\hel{k}^b \times \el{b} \to \el{k} \cup \set{\bot}$
such that for every $\h{w}\in\hel{k}^b$ and for every $i\in\el{b}$
holds $\alpha\of{\h{w}, i} \neq \bot$ if and only if $\h{w}_i = \h{1}.$ 
There is a permutation $\phi:\hel{k}^b \to \el{k}^b$
such that for every $\h{w}\in\hel{k}^b$ and for every $i\in\el{b}$
if $\h{w}_i = \h{1}$ then $\phi\of{\h{w}}_i = \alpha\of{\h{w}, i}.$
\end{lemma}

\begin{proof}
We will use an induction on $b.$

For $b = 0$ we have $\phi\of{\varepsilon} = \varepsilon.$

For $b > 0$ we can define functions
$\alpha_1, \alpha_2, \ldots, \alpha_k:
  \hel{k}^{b - 1} \times \el{b - 1} \to \el{k} \cup \set{\bot}$
such that for every $x\in\el{k}$ every $\h{w}\in\hel{k}^{b - 1}$
and every $i\in\el{b - 1}$
we put $\alpha_x\of{\h{w}, i} = \alpha\of{\h{x}\h{w}, i + 1}.$

From the inductive hypothesis for $b - 1$ used for every function of
$\alpha_1, \alpha_2, \ldots, \alpha_k$
we got the permutations
$\phi_1, \phi_2, \ldots, \phi_k: \hel{k}^{b - 1} \to \el{k}^{b - 1}$
such that for every $x\in\el{k}$ for every $\h{w}\in\hel{k}^{b - 1}$
and for every $i\in\el{b - 1}$
we have that if $\h{w}_i = \h{1}$ then
$\phi_x\of{\h{w}}_i = \alpha_x\of{\h{w}, i}
  = \alpha\of{\h{x}\h{w}, i + 1}.$

Now we can use Lemma~\ref{lem:bijection-merge} to merge the permutations
$\phi_1, \phi_2, \ldots, \phi_k$
using a function $\rho:\hel{k}^{b - 1} \to \el{k}$ such that
$\rho\of{\h{w}} = \alpha\of{\h{1}\h{w}, 1}$
for every $\h{w}\in\hel{k}^{b - 1}.$
We obtain one permutation
$\phi:\hel{x}^b \to \el{x}^b$
such that by Lemma~\ref{lem:bijection-merge} (i)
for every $\h{x}\in\hel{k},$ for every $\h{w}\in\hel{k}^{b - 1}$
and for every $i\in\el{b - 1}$ we have that
$\phi\of{\h{x}\h{w}}_{i + 1} = \phi_x\of{\h{w}}_i$
so if $\h{w}_i = \h{1}$ then
$\phi\of{\h{x}\h{w}}_{i + 1} = \phi_x\of{\h{w}}_i
  = \alpha\of{\h{x}\h{w}, i + 1}.$
Also by Lemma~\ref{lem:bijection-merge} (ii),
for every $\h{w}\in\hel{k}^{b - 1}$ we have that
$\phi\of{\h{1}\h{w}}_1 = \rho\of{\h{w}} = \alpha\of{\h{1}\h{w}, 1}.$
So for every $\h{w}\in\hel{k}^b$ and for every $i\in\el{b}$
we have that if $\h{w}_i = \h{1}$ then
$\phi\of{\h{w}}_i = \alpha\of{\h{w}, i},$
as required.
\end{proof}

\begin{figure}[t]
\begin{center}
  \def\svgwidth{5in}
  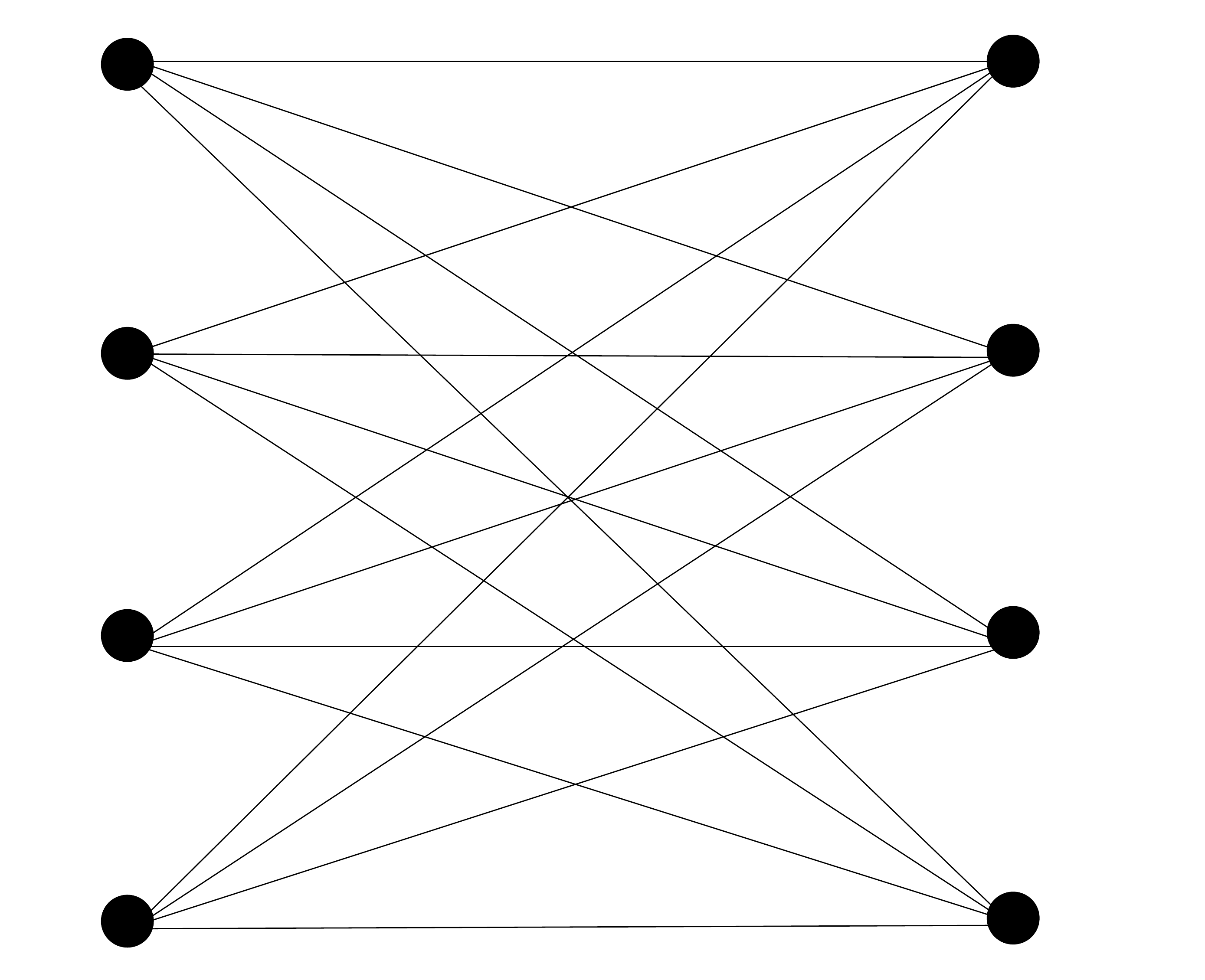
  \caption{
    The graph encoding an $f$-family for $f:\feln{4}{2}.$
    The lower indices $(1), (2), (3)$ and $(4)$
    are added to indicate the correspondence between
    the occurrences of $\h{1}$ and the elements of $\el{n}$
    (the first argument of the function $f$).
    The lower indices $(i)$ and $(ii)$ are added to indicate the correspondence
    between the second argument of the function $f$ and the position in the
    (two element) sequence $\langle \cdot, \cdot \rangle.$
  }
  \label{fig:k44}
\end{center}
\vspace{-1cm}
\end{figure}

\section{Omitted Proofs from Section~\ref{sec:CA}}

\begin{lemma}[Lemma~\ref{lem:ca-merge}]
For every $(u, v)$-spanned instance $I_1=(V_1, d_1, s)$ and
$(w, z)$-spanned instance $I_2=(V_2, d_2, s)$
of \CA
there is a $(\set{u,w}, \set{v,z})$-spanned instance $I=(V_1 \cup V_2, d, s)$
of \CA
such that
\begin{enumerate}[(i)]
  \item for every YES-coloring $c$ of $I$
    the coloring $c\mid_{V_1}$ is a YES-coloring of $I_1$
    and the coloring $c\mid_{V_2}$ is a YES-coloring of $I_2,$
  \item for every YES-coloring $c_1$ of $I_1$
    and every YES-coloring $c_2$ of $I_2$
    such that $c_1\of{u} = c_2\of{w},$ $c_1\of{v} = c_2\of{z}$ and
    for every $x\in V_1 \cap V_2$ we have $c_1\of{x} = c_2\of{x}$
    there exists a YES-coloring $c$ of $I$
    such that $c\mid_{V_1} = c_1$ and $c\mid_{V_2} = c_2.$
\end{enumerate}
\end{lemma}

\begin{proof}
Let $B = \set{u, w} \times \set{v, z} \cup \set{v, z} \times \set{u, w}$
and let
\[
  d(x, y) =
  \begin{cases}
    s - 1 & \mbox{if } (x, y)\in B\\
    \max\{d_1(x, y), d_2(x, y)\} & \mbox{if } x, y\in V_1 \cap V_2\\
    d_1(x, y) & \mbox{if } x, y \in V_1 \mbox{ and } x,y\not\in V_1\cap V_2\\
    d_2(x, y) & \mbox{if } x, y \in V_2 \mbox{ and } x,y\not\in V_1\cap V_2\\
    0 & \mbox{otherwise}.
  \end{cases}
\]

Our instance is $(\set{u, w}, \set{v, z})$-spanned because for all the pairs
in $B$ we set the minimum allowed distance to at least $s - 1.$

Note that for $i=1,2,$ for every $x,y\in V_i$ we have $d(x, y) \geq d_i(x, y).$
Hence every
proper coloring $c$ of $I$ has the property that $c\mid_{V_1}$
is a proper coloring of $I_1$ and $c\mid_{V_2}$ is a proper coloring of $I_2.$
Also the maximum allowed spans of $I, I_1, I_2$ are the same,
so for every YES-coloring $c$ of $I$ coloring $c\mid_{V_1}$ is a YES-coloring
of $I_1$ and $c\mid_{V_2}$ is a YES-coloring of $I_2.$
Hence (i) is clear.

For (ii), consider
a YES-coloring $c_1$ of $I_1$ and a YES-coloring $c_2$ of $I_2$
such that $c_1\of{u} = c_2\of{w}$ and $c_1\of{v} = c_2\of{z}$ and
such that for every $x\in V_1 \cap V_2$ we have $c_1\of{x} = c_2\of{x}.$
Then we define a coloring $c$
\[
  c\of{x} = \begin{cases}
    c_1\of{x} & \mbox{if } x\in V_1\\
    c_2\of{x} & \mbox{if } x\in V_2.\\
  \end{cases}
\]

We know that $c_1\of{u} = c_2\of{w}$ and $c_1\of{v} = c_2\of{z}$
so all the vertices of $V_1 \cup V_2$ have colors
between $c_1\of{u}$ and $c_1\of{v},$
i.e., the $\myspan$ of $c$ is at most $s$ as required.
It is straightforward to check that $c$ is a proper coloring.
\end{proof}

\begin{lemma}[Lemma~\ref{lem:ca-extend}]
For every $(v_L, v_R)$-spanned instance $I=(V, d, s)$ of \CA and
for every numbers $l, r \in \N$ there exists
a $(w_L, w_R)$-spanned instance $I'=(V \cup \set{w_L, w_R}, d', l + s + r)$
such that
\begin{enumerate}[(i)]
  \item for every YES-coloring $c$ of $I$
    there is a YES-coloring $c'$ of $I'$ such that $c'\mid_V = c,$
  \item for every YES-coloring $c'$ of $I'$ such that
    $c'\of{w_L} \leq c'\of{w_R}$ we have that 
    \begin{itemize}
      \item a coloring $c'\mid_V$ is a YES-coloring of $I,$
      \item $c'\of{v_L} = c'\of{w_L} + l$ and
        $c'\of{v_R} = c'\of{w_R} - r.$
    \end{itemize}
\end{enumerate}
\end{lemma}

\begin{proof}
We assume that $w_L, w_R\not\in V.$ We put
\[
  d'\of{x, y} =
  \begin{cases}
    l + s - 1 + r & \mbox{for } \set{x, y} = \set{w_L, w_R}\\
    l & \mbox{for } \set{x, y} \cap \set{w_L, w_R} = \set{w_L}\\
    r & \mbox{for } \set{x, y} \cap \set{w_L, w_R} = \set{w_R}\\
    d\of{x, y} &\mbox{for } x,y \in V.\\
  \end{cases}
\]
It is straightforward to check that $d'$ satisfies (i) and (ii).
\end{proof}

\begin{lemma}[Lemma~\ref{lem:matchings-to-ca}, proof of {\bf Claim 6}]
Let $G = (V_1 \cup V_2, E, w)$ be a weighted full bipartite graph
with nonnegative weights and such that $\pow{V_1} = \pow{V_2}.$
Let $n = \pow{V_1},$
$m = \max_{e\in E} w(e),$
$M = n \cdot m + 1,$
$l = \p{4n - 1} \cdot M$ and
$s = (8n - 1) \cdot M.$
There exists a $(v_L, v_R)$-spanned instance $I=(V, d, s)$ of \CA
with $\pow{V} = O\of{n}$ and
such that for some vertex $v_M \in V,$
\begin{enumerate}[(i)]
  \item for every YES-coloring $c$ of $I$
    such that $c\of{v_L} \leq c\of{v_R}$
    there exists a perfect matching $M_G$ in $G$ such that
    $c\of{v_M} = c\of{v_L} + l + w\of{M_G}$ and
  \item for every perfect matching $M_G$ in $G$ there exists a YES-coloring
    $c$ of $I$ such that $c\of{v_L} \leq \of{v_R}$ and
    $c\of{v_M} = c\of{v_L} + l + w\of{M_G}.$
\end{enumerate}
\end{lemma}

\begin{proof}
{\bf Claim 6}
Let $\pi:\el{n}\to\el{n}$ be any permutation and
$M_{\pi} = \set{\vv{1}{i} \vv{2}{\pi\of{i}} : i\in\el{n}}$ be 
the corresponding perfect matching in $G.$
There is exactly one YES-coloring $c$ such that
$\pi_c = \pi.$
Moreover $c(v_M) = c(v_L) + l + w(M_{\pi}).$

{\em Proof of the claim:}
Let us consider a sequence of the vertices
$$v_1, a_{\pi\of{1}}, v_2, w_1, v_3, a_{\pi\of{2}}, v_4, w_2, v_5 \ldots,
  v_{2n - 1}, a_{\pi\of{n}}, v_{2n}, w_n,$$ $$
  v_{2n + 1}, b_{\pi\of{1}}, v_{2n + 2}, w_{n + 1},
  v_{2n + 3}, b_{\pi\of{2}}, v_{2n + 4}, w_{n + 2}, v_{2n + 5}\ldots
  v_{4n - 1}, b_{\pi\of{n}}, v_{4n}$$
and the coloring $c$ implied by the minimum distances of pairs of consecutive
elements in this sequence, i.e.,
$c\of{v_1} = 1,$
$c\of{a_{\pi\of{1}}} = c\of{v_1} + d\of{v_1, a_{\pi\of{1}}},$
$c\of{v_2} = c\of{a_{\pi\of{1}}} + d\of{a_{\pi\of{1}}, v_2},$
$c\of{w_1} = c\of{v_2} + d\of{v_2, w_1},$
$c\of{v_3} = c\of{w_1} + d\of{w_1, v_3},$
\ldots,
$c\of{v_{\maxv}} = c\of{b_{\pi\of{n}}} + d\of{b_{\pi\of{n}}, v_{\maxv}}.$
We need to check that all the minimum distance constraints $d$ are satisfied
and that the span of this coloring is nor greater than $s.$

Note that for every $i\in\elv$ and for every vertex $x\in V$ such that
$v_i \ne x$ we have $d(x, v_i) \ge M.$
Therefore for every $i\in\el{\maxv - 1}$ we have
$c(v_{i + 1}) - c(v_i) = (c(v_{i + 1}) - c(x)) + (c(x) - c(v_i))
  = d(x, v_{i + 1}) + d(v_i, x) \ge 2M$
where $x$ is the vertex separating $v_i$ and $v_{i + 1}$ in the sequence.
Thus for every $i, j \in \elv$ we have
$|c(v_i) - c(v_j)| \ge |i - j| \cdot 2M$
so if $\{i, j\} \ne \{1, \maxv\}$ then $|c(v_i) - c(v_j)| \ge d(v_i, v_j).$
Hence also for every $i\in\elw$ and $j\in\elv$ we have
$|c(w_i) - c(v_j)| = |c(w_i) - c(v_k)| + |c(v_k) - c(v_j)|
\ge M + |k - j| \cdot 2M$
where in case that $j \le 2i$ we have $k = 2i$ and in this case
$M + |k - j| \cdot 2M = |4i -2j + 1| \cdot M$
and in case that $j > 2i$ we have $k = 2i + 1$ and in this case
$M + |k - j| \cdot 2M = |2j - 4i - 1| \cdot M = |4i -2j + 1| \cdot M$
so in both cases
$|c(w_i) - c(v_j)| \ge |4i - 2j + 1| \cdot M = d(w_i, v_j).$
We will check the distance between $v_L = v_1$ and $v_R = v_{\maxv}$ later.

For every $i\in\ela$ and vertex $a_{\pi(i)}$ the closest vertex $v_j$
to the left is $v_{2i - 1}$ and to the right is $v_{2i}.$
They are immediate neighbours of $a_{\pi(i)}$ in the sequence so
from the definition of $c$ we have
$|c(a_{\pi(i)}) - c(v_{2i - 1})| = d(v_{2i - 1}, a_{\pi(i)})$
and
$|c(v_{2i}) - c(a_{\pi(i)})| = d(a_{\pi(i)}, v_{2i}).$
Note that for every $j\in\el{2n}$ we have $d(a_{\pi(i)}, v_j)\le 2M$
and then for every $j\in\el{2i - 2}$ we have
$ |c(a_{\pi(i)}) - c(v_j)|
= (c(v_{2i - 1}) - c(v_j)) + (c(a_{\pi(i)}) - c(v_{2i - 1}))
\ge 2M + M
> d(a_{\pi(i)}, v_{2i - 1}).$
Similarly for every $2i + 1 \leq j \leq 2n$ we have
$ |c(a_{\pi(i)}) - c(v_j)|
= ((c(v_{2i}) - c(a_{\pi(i)})) + (c(v_j) - c(v_{2i}))
\ge M + 2M
> d(a_{\pi(i)}, v_j).$
For every $2n + 1 \leq j \leq \maxv$
we have
$ |c(v_j) - c(a_{\pi(i)})|
  = (c(v_{2i}) - c(a_{\pi(i)})) + (c(v_{2n}) - c(v_{2i}))
      + (c(v_j) - c(v_{2n}))
  \ge M + 0 + (j - 2n) \cdot 2M
  = d(a_{\pi(i)}, v_j).$
Because $\pi$ is a permutation thus we obtain that
for every $i\in\ela$ and for every $j\in\elv$ we have
$|c(a_i) - c(v_j)| \ge d(a_i, v_j).$

For every $i\in\ela$ and $j\in\elw$ there is at least one vertex $v_k$
with color between the colors $c(a_i)$ and $c(w_j)$ so
$ |c(a_i) - c(w_j)|
= |c(a_i) - c(v_k)| + |c(v_k) - c(w_j)|
\ge 2M
= d(a_i, w_j).$
For every $i, j \in\ela$ such that $\pi^{-1}(i) < \pi^{-1}(j)$
there are at least two vertices
$v_k, v_{k + 1}$ with colors
$c(a_i) \leq c(v_k) \leq c(v_{k + 1}) \leq c(a_j).$
Therefore
$ |c(a_j) - c(a_i)|
= (c(v_k) - c(a_i)) + (c(v_{k + 1}) - c(v_k)) + (c(a_j) - c(v_{k + 1}))
\ge M + 2M + M = 4M = d(a_i, a_j).$

Similarly we can check that for every $i\in\elb$ and $j\in\elv$ we have
$|c(b_i) - c(v_j)| \ge d(b_i, v_j),$
that for every $i\in\elb$ and $j\in\elw$ we have
$|c(b_i) - c(w_j)| \ge d(b_i, w_j)$
and for every $i,j\in\elb$ such that $i\ne b$ we have
$|c(b_i) - c(b_j)| \ge d(b_i, b_j).$

We need also to check that for every $i\in\ela$ we have
$|c(a_i) - c(b_i)| \ge n\cdot 4M = d(a_i, b_i).$
Indeed
$|c(b_i) - c(a_i)|
= (c(v_{2i}) - c(a_i)) + (c(v_{2n + 2i - 1}) - c(v_{2i}))
  + (c(b_i) - c(v_{2n + 2i - 1}))
\ge M + (2n - 1) \cdot 2M + M
= n\cdot 4M
= d(a_i, b_i).$

Now we are going to deal with the distances between $v_L, v_M$ and $v_R.$
The sum of the minimum color distances of neighbouring elements
in the prefix of our sequence:
$$v_1, a_{\pi\of{1}}, v_2, w_1, v_3, a_{\pi\of{2}}, v_4, w_2, v_5 \ldots,
  v_{2n - 1}, a_{\pi\of{n}}, v_{2n}, w_n, v_{2n + 1}$$
is exactly $2n \cdot 2M + w\of{M_{\pi}}.$
The sum of the minimum color distances of neighbouring elements
in the sufix of our sequence:
$$v_{2n + 1}, b_{\pi\of{1}}, v_{2n + 2}, w_{n + 1},
  v_{2n + 3}, b_{\pi\of{2}}, v_{2n + 4}, w_{n + 2}, v_{2n + 5}\ldots
  v_{4n - 1}, b_{\pi\of{n}}, v_{4n}$$
is exactly $\p{2n - 1} \cdot 2M + n\cdot m - w\of{M_{\pi}}.$
So the total sum for the whole sequence is
exactly $\p{4n - 1} \cdot 2M + n \cdot m = s - 1$
and it does not depend on the permutation $\pi.$
Therefore $|c(v_R) - c(v_L)| = s - 1 = d(v_L, v_R).$
This was the last constraint to check
and hence we have shown that $c$ is proper.
On the other hand the $\myspan$ of $c$ is $s$ so $c$ is a YES-coloring.
Moreover we have
$c\of{v_M} = c\of{v_L} + \p{4n - 1} \cdot M + w\of{M_{\pi_c}}
  = c\of{v_L} + l + w\of{M_{\pi_c}}.$
Note that all the distances of pairs of consecutive elements of (the whole)
sequence are tight, i.e.,
these distances are equal to the minimum allowed distances for these pairs
of the vertices and therefore we cannot decrease any of these distances.
On the other hand the $\myspan$ of $c$ is maximum so we cannot increase any
of these distances without exceeding the maximum span or
violating some of the constraints provided by $d$.
Therefore $c$ is the only one YES-coloring for which the colors of the
vertices of this sequence are increasing.
Hence $c$ is the only one YES-coloring such that $\pi_c = \pi.$
This ends the proof of the claim.
\end{proof}

\begin{lemma}[Lemma~\ref{lem:matchings-to-ca}]
Let $G = (V_1 \cup V_2, E, w)$ be a weighted full bipartite graph
with nonnegative weights and such that $\pow{V_1} = \pow{V_2}.$
Let $n = \pow{V_1},$
$m = \max_{e\in E} w(e),$
$M = n \cdot m + 1,$
$l = \p{4n - 1} \cdot M$ and
$s = (8n - 1) \cdot M.$
There exists a $(v_L, v_R)$-spanned instance $I=(V, d, s)$ of \CA
with $\pow{V} = O\of{n}$ and
such that for some vertex $v_M \in V,$
\begin{enumerate}[(i)]
  \item for every YES-coloring $c$ of $I$
    such that $c\of{v_L} \leq c\of{v_R}$
    there exists a perfect matching $M_G$ in $G$ such that
    $c\of{v_M} = c\of{v_L} + l + w\of{M_G}$ and
  \item for every perfect matching $M_G$ in $G$ there exists a YES-coloring
    $c$ of $I$ such that $c\of{v_L} \leq \of{v_R}$ and
    $c\of{v_M} = c\of{v_L} + l + w\of{M_G}.$
\end{enumerate}
\end{lemma}

\begin{proof}
Let use Lemma~\ref{lem:matchings-to-ca} on graph $G_1$
to obtain a $\p{\vv{1}{L}, \vv{1}{R}}$-spanned
\CA instance $I_1 = (V_1, d_1, s_1)$ with
$l_1 = O\of{n_1^2 m_1},$
$s_1 = 2l_1 + n_1 \cdot m_1 = O\of{n_1^2 m_1}$
and with the vertex $\vv{1}{M}$
(as in the statement of Lemma~\ref{lem:matchings-to-ca}).
The number of the vertices in $V_1$ is $O\of{n_1}.$

Similarly, let $I_2 = (V_2, d_2, s_2)$ be a
$\p{\vv{2}{L}, \vv{2}{R}}$-spanned \CA instance with
$l_2 = \O\of{n_2^2 m_2},$
$s_2 = 2l_2 + n_2 \cdot m_2 = O\of{n_2^2 m_2}$
and with the vertex $\vv{2}{M}$
obtained from Lemma~\ref{lem:matchings-to-ca} from graph $G_2.$
The number of the vertices in $V_2$ is $O\of{n_2}.$

Let us identify vertices $\vv{1}{M}$ and $\vv{2}{M},$ i.e.,
$\vv{1}{M} = \vv{2}{M} = v_M$ and
$V_1 \cap V_2 = \set{v_M}.$
And let $l_{\max} = \max\set{l_1, l_2} = O\of{n_1^2 m_1 + n_2^2 m_2}$ and
$s = l_{\max} + \max\set{s_1 - l_1, s_2 - l_2} = O\of{n_1^2 m_1 + n_2^2 m_2}.$

Our span will be $s$.
Note that then every edge with a weight greater than $s - 1$ forces
that our instance is a NO-instance.
So if we have an edge with a weight greater that $s$ we can replace it
with the same edge with but a weight equal to $s$ and the istance will be still
a NO-instance.
Therefore weights of
all our edges will be bounded by $O\of{n_1^2 m_1 + n_2^2 m_2}.$

We can use Lemma~\ref{lem:ca-extend} with $l = l_{\max} - l_1$
and with $r = s - \p{l_{\max} + s_1 - l_1}$
for extending the instance $I_1$
into a $\p{\ww{1}{L}, \ww{1}{R}}$-spanned
instance $I_1'\of{V_1' = V_1 \cup \set{\ww{1}{L}, \ww{1}{R}}, d_1', s}$ of \CA.

For every YES-coloring $c_1'$ of $I_1'$ we know that $c_1'\mid_{V_1}$
is a YES-coloring of $I_1$ and for every YES-coloring $c_1$ of $I_1$
there exists a YES-coloring $c_1'$ of $I_1'$ such that $c_1'\mid_{V_1} = c_1$
so from the properties of $I_1$
(obtained from Lemma~\ref{lem:matchings-to-ca})
we know that
\begin{itemize}
  \item for every YES-coloring $c_1'$ of $I_1'$
    such that $c_1'\of{\ww{1}{L}} \leq c_1'\of{\ww{1}{R}}$
    there exists a perfect matching $M_1$ in $G_1$ such that
    $c_1'\of{v_M} = c_1'\of{\ww{1}{L}} + l_{\max} + w_1\of{M_1}$ and
  \item for every perfect matching $M_1$ in $G_1$
    there exists a YES-coloring $c_1'$ of $I_1'$ such that
    $c_1'\of{\ww{1}{L}} \leq \of{\ww{1}{R}}$ and
    $c_1'\of{v_M} = c_1'\of{\ww{1}{L}} + l_{\max} + w_1\of{M_1}.$
\end{itemize}

Similarly we can use Lemma~\ref{lem:ca-extend} with $l = l_{\max} - l_2$
and with $r = s - \p{l_{\max} - l_2 + s_2}$
for extending the instance $I_2$
into a $\p{\ww{2}{L}, \ww{2}{R}}$-spanned
instance $I_2'\of{V_2' = V_2 \cup \set{\ww{2}{L}, \ww{2}{R}}, d_2', s}$ of \CA
such that
\begin{itemize}
  \item for every YES-coloring $c_2'$ of $I_2'$
    such that $c_2'\of{\ww{2}{L}} \leq c_2'\of{\ww{2}{R}}$
    there exists a perfect matching $M_2$ in $G_2$ such that
    $c_2'\of{v_M} = c\of{\ww{2}{L}} + l_{\max} + w\of{M_2}$ and
  \item for every perfect matching $M_2$ in $G_2$
    there exists a YES-coloring $c_2'$ of $I_2'$ such that
    $c_2'\of{\ww{2}{L}} \leq \of{\ww{2}{R}}$ and
    $c_2'\of{v_M} = c_1'\of{\ww{2}{L}} + l_{\max} + w_1\of{M_1}.$
\end{itemize}

Now we can use Lemma~\ref{lem:ca-merge} to merge the instances $I_1'$
and $I_2'$ into a one
$\p{\set{\ww{1}{L}, \ww{2}{L}}, \set{\ww{1}{R}, \ww{2}{R}}}$-spanned
instance $I'=\of{V_1' \cup V_2', d, s}.$
A simplified picture of the obtained instance can be found in
Figure~\ref{fig:matchings-to-ca}.
Note that $V_1' \cap V_2' = \set{v_M}$ so
\begin{itemize}
  \item for every YES-coloring $c$ of $I'$ such that
    $c\of{\ww{1}{L}} \leq c\of{\ww{1}{R}}$ there exist perfect matchings
    $M_1$ in $G_1$ and $M_2$ in $G_2$ such that
    $c\of{v_M} = c\of{\ww{1}{L}} + l_{\max} + w_1\of{M_1}
      = c\of{\ww{1}{L}} + l_{\max} + w_2\of{M_2},$
    so $w_1\of{M_1} = w_2\of{M_2}$ and
  \item for every two perfect matchings $M_1$ in $G_1$
    and $M_{G_2}$ in $G_2$ such that $w_1\of{M_1} = w_2\of{M_2}$
    there is a YES-coloring $c$ of $I$ such that
    $c\of{\ww{1}{L}} \leq c\of{\ww{1}{R}}$ and
    $c\of{v_M} = c\of{\ww{1}{L}} + l_{\max} + w_1\of{M_1}.$
\end{itemize}

Therefore the \CA instance $I'$ has a YES-coloring if and only if there are
two perfect matchings $M_1$ in $G_1$ and $M_2$ in $G_2$
such that $w_1\of{M_1} = w_2\of{M_2}.$
By Lemma~\ref{lem:ca-merge} and Lemma~\ref{lem:ca-extend}
we know that $I'$ has $O\of{n_1 + n_2}$ vertices.
\end{proof}

\begin{corollary}[Corollary~\ref{cor:ca-nloglogl}]
There is no algorithm solving \CA in
$2^{n \cdot o(\log\log\ell)} \poly{r}$
where $n$ is the number of the vertices
and $r$ is the bit size of the instance
unless ETH fails.
\end{corollary}

\begin{proof}
For a given instance of \ThreeSAT with $n$ variables and $m$ clauses we
use the reduction from Theorem~\ref{thm:ThreeSAT-to-CMW} to obtain an instance
of \CMW with $|V_1| = O\of{\ndivlog{n}},$ $|V_2| = O\of{\ndivlog{m}}$ 
and the maximum matching weights bounded by $2^{3m}.$

Then we use the reduction from Lemma~\ref{lem:CMW-to-CA} to obtain an instance
of \CA with
\newcommand{\nprime}{\frac{n + m}{\log(n + m)}}
$$n'
= O\of{\ndivlog{n} + \ndivlog{m}}
= O\of{\nprime}$$
vertices and the weights
on the edges bounded by
$$\ell = O\of{\p{\ndivlog{n} + \ndivlog{m}}^2 \cdot 2^{3m}}.$$
Then
$\log\ell = O(n + m)$ and
$r = O((n')^2 \cdot \log\ell) = O((n + m)^3).$

Let us assume that there is an algorithm solving \CA in 
$2^{n \cdot o(\log\log\ell)} \poly{r}$-time
then we can solve our instance in time
$$ 2^{O\of{\nprime} \cdot o(\log O(n + m))} \poly{O((n + m)^3)}
= 2^{O\of{\nprime} \cdot o(\log(n + m))} \poly{n + m}
= 2^{o(n + m)}$$
which contradicts ETH by Corollary~\ref{cor:ThreeSAT-hardness}.
\end{proof}

\end{document}